\documentclass[a4paper,11pt]{article}

\usepackage[margin=1in]{geometry}
\usepackage{setspace}
\usepackage{amsmath}
\usepackage{amssymb}
\usepackage{amsthm}
\usepackage{mathrsfs}
\usepackage{esint}
\usepackage{hyperref}

\newtheorem{theorem}{Theorem}[section]

\hypersetup{
    hidelinks,
    pdftitle={An Electrodynamic Lifting of the Toda Monopoles in Curved 3+1 Spacetime},
    pdfauthor={Hanwen Liu}
}

\begin{document}

\title{\texorpdfstring{\textbf{An Electrodynamic Lifting of the Toda Monopoles in Curved 3+1 Spacetime}}{An Electrodynamic Lifting of the Toda Monopoles in Curved 3+1 Spacetime}}
\author{Hanwen Liu}
\date{}

\maketitle

\begin{abstract}
Every solution of the hyperbolic Toda equation determines a classical Abelian monopole on a (2+1)-dimensional Einstein--Weyl space. We determine when its connection is also a source-free electromagnetic potential on a curved Lorentzian 3+1 spacetime with a spacelike Killing field. On the generic branch, Maxwell's equation classifies the admissible Killing-fibre scale and horizontal twist through one characteristic of the Toda field. A separated Liouville branch gives purely magnetic fields with considerably freer fibre geometry. The round spherical solution yields a monopole of unit charge with conserved flux on a contracting screen. The background metric here is neither assumed to be Einstein nor self-dual.
\end{abstract}

\begin{center}
\end{center}

%\tableofcontents
\onehalfspacing
\raggedbottom

\section*{Acknowledgement}

The author is deeply grateful to Alexander Veselov and Evgeny Ferapontov for encouragement and inspiring discussions. Special thanks should also go to Michael Lau for careful proofreading.

\section{Introduction and Background}

The continuous Toda equation entered self-dual gravity through the work of Boyer and Finley \cite{BoyerFinley}. Its usual Riemannian form is elliptic; the equation considered here,
$$
\ddot{q}+\dot{q}^2=e^{-q}(q_{xx}+q_{yy}),
$$
is its hyperbolic or wave form. It is also called the Boyer--Finley equation and the dispersionless 2-dimensional Toda equation. The continual Lie-algebraic and hierarchy interpretations were developed in \cite{SavelievVershik,TakasakiTakebe,Strachan}, while the wave-form inverse problem was studied in \cite{ManakovSantini}. These are different aspects of the same integrable equation, rather than different equations concealed by the notation.

The geometric meaning of the equation is classical. Jones and Tod related self-dual conformal four-manifolds with a conformal Killing vector to Einstein--Weyl three-manifolds and Abelian generalized monopoles \cite{JonesTod}; Ward then placed the Toda equation explicitly in Einstein--Weyl geometry \cite{Ward}. The resulting correspondence, its Toda realizations and its families of self-dual metrics were developed further in \cite{Calderbank,CalderbankPedersen,CalderbankTod}. In this setting, the derivative $\dot q$, together with a multiple of $d^cq$, is the canonical Toda monopole. The phrase \emph{Toda monopole} will always mean this Abelian Jones--Tod monopole, and not a non-Abelian Bogomol'nyi–Prasad–Sommerfield monopole whose Nahm data happen to satisfy affine Toda equations.

There is also a classical passage from monopoles to 4-dimensional differential geometry. Abelian monopoles already drive the Gibbons--Hawking construction \cite{GibbonsHawking}, and the Jones--Tod correspondence uses generalized monopole data to reconstruct a self-dual conformal structure. Moreover, a second generalized monopole on the quotient gives a self-dual electromagnetic field on the reconstructed 4-manifold \cite{CalderbankPedersen}; related links between self-dual geometry and Maxwell curvature appear in \cite{BoyerPlebanski}. Thus, neither the Toda monopole nor the broad principle that monopoles may lift to electromagnetic fields is new.

The problem considered here is more specific. We prescribe a real Lorentzian spacetime with one spacelike Killing direction, regard its spin connection as an external electromagnetic potential, and ask exactly when its curvature is source-free. The Killing-fibre connection remains independent of the Toda-monopole connection, and the background is not required to be Einstein, self-dual or conformally self-dual. General invariant electromagnetic fields on spacetimes with a non-null Killing field have long been studied \cite{Racz}; what is isolated below is the exact Toda-generated subclass and the local classification of the fibre data that support it.

The first theorem gives this classification wherever $\dot q$ and its horizontal differential are nonzero. In this case, the reciprocal characteristic $s:=t-1/\dot q$, as a shifted time variable, determines the fibre scale, while the time derivative of the horizontal twist is constrained to follow the level sets of $\dot q$. A short nonseparated example shows that this branch contains mixed electric and magnetic fields. The second theorem treats the complementary separated branch. Its screen geometry obeys the Liouville equation, its field is purely magnetic, and its fibre geometry has substantially greater freedom. The final example globalizes the curvature over a round sphere and realizes the usual unit magnetic charge on an evolving background.

The precise Lorentzian test-field lift and its radius--twist classification appear not to have been recorded previously. This is the novelty claimed here. The Toda equation, its integrability, its Einstein--Weyl interpretation and its canonical monopole are used as classical input. Nearby Riemannian Einstein--Maxwell constructions based on Toda geometry, such as those presented in \cite{Araneda}, impose different field equations and belong to a different lifting problem.

All statements are local unless a global holographic screen surface is explicitly introduced. Section~2 fixes the local conventions and proves the lifting results, Section~3 gives the spherical monopole, and Section~4 records the scope and natural extensions of the construction.

\section{The Main Results on Electrodynamic Liftings}

As per usual, by an electromagnetic potential, we mean a $\operatorname{U}(1)$ connection 1-form whose curvature $F$ obeys the Maxwell equation $dF=d\star F=0$, where $\star:=\star_4$ is the Hodge star of the generally curved background 3+1 spacetime.

We work on a coordinate patch on which the synchronous normal form below exists:

Let $g$ be a 3+1 spacetime metric admitting a spacelike Killing vector field, of which horizontal screen is integrable and the screen shear vanishes, so in particular, locally there exist smooth scalar functions $q:=q(t;x,y)$ and $A:=A(t;x,y),B:=B(t;x,y),C:=C(t;x,y)$ with $C>0$, such that 
\begin{equation}\label{metric}
g=-c^2dt^2+e^q(dx^2+dy^2)+(Adx+Bdy+Cdz)^2
\end{equation}
where 
$$
\frac{\partial}{\partial x}\left(\frac{B}{C}\right)-\frac{\partial}{\partial y}\left(\frac{A}{C}\right)=0,
$$
and $(x^0,x^1,x^2,x^3):=(t,x,y,z)$ here is the 4-position in the universe. 

Note that, since 
$$
\eta(t):=\frac{A(t)dx+B(t)dy}{C(t)}
$$
is a closed spatial 1-form at each fixed time $t$, the spin connection of $g$ reduces to 
$$
\omega:=d^cq=\frac{1}{2}(q_ydx-q_xdy),
$$
where we use the standard $d^c$ operator in the $xOy$-plane identified with $\mathbb{C}\equiv\mathbb{R}^2$.

From now on, for simplicity, we shall normalize the units so that the speed of light is $c=1$.

We first make the relation with the classical monopole explicit. We write $p:=\dot q$, and denote by $\star_3$ the Hodge star operator of the 2+1 Lorentzian metric $-dt^2+e^q(dx^2+dy^2)$ with orientation $dt\wedge dx\wedge dy$. A direct differentiation shows that the Toda equation~(\ref{Toda}) is equivalent to
\begin{equation}\label{Toda_monopole}
2d\omega=\star_3(dp+p^2dt).
\end{equation}
The Weyl form in the Jones--Tod convention is $2pdt$, so equation~(\ref{Toda_monopole}) is precisely its generalized Abelian monopole equation for the pair $(p,2\omega)$. The factor $2$ merely records our normalization of the differential $d^c$. In the present problem, the connection $\omega$ is lifted as the electromagnetic potential, whereas the Killing form $\eta+dz$ is determined independently.

\begin{theorem}\label{main_theorem}
Let $g=-dt^2+e^q(dx^2+dy^2)+C^2(\eta+dz)^2$ be the metric in (\ref{metric}). Assume the conformal factor $q$ is a local solution to the Toda equation
\begin{equation}\label{Toda}
\ddot{q}+\dot{q}^2=e^{-q}(q_{xx}+q_{yy})
\end{equation}
with $\dot{q}^2(\dot{q}_x\dot{q}_x+\dot{q}_y\dot{q}_y)>0$. Then, the spin connection $\omega=d^cq$ of the spacetime $g$ is an electromagnetic potential if and only if $C=\exp(f(t-1/\dot{q})-2\log|\dot{q}|)$ and 
$$
\dot{\eta}=\frac{\partial H(t,\dot{q})}{\partial x}dx+\frac{\partial H(t,\dot{q})}{\partial y}dy
$$
for some smooth functions $f=f(t)$ and $H=H(t,p)$.
\end{theorem}
\begin{proof}
Denote by $\star$ the Hodge star of $g$, and set $p:=\dot q$. A direct computation gives
$$
F_\omega:=d\omega
=\frac{1}{2}dt\wedge(p_ydx-p_xdy)
-\frac{1}{2}(q_{xx}+q_{yy})dx\wedge dy.
$$
With respect to the induced orientation $dt\wedge dx\wedge dy\wedge dz$, we also obtain
$$
\star F_\omega
=-\frac{C}{2}(p_xdx+p_ydy+e^{-q}(q_{xx}+q_{yy})dt)\wedge(\eta+dz),
$$
the dual 2-form of the Faraday tensor.

By the Toda equation (\ref{Toda}), we have that
$$
p_xdx+p_ydy+e^{-q}(q_{xx}+q_{yy})dt
=dp+p^2dt=p^2d(t-1/p).
$$
Since $\eta(t)$ is a closed spatial 1-form at each fixed time $t$, separating the terms containing $\eta+dz$ from the horizontal terms gives
$d\star F_\omega=0$
if and only if the system~($*$)
$$
\begin{cases}
d(Cp^2d(t-1/p))=0\\
p_xdx+p_ydy\wedge\dot\eta=0
\end{cases}
$$
is satisfied.
Here, the second equation in the system~($*$) is understood as
$$
(p_xdx+p_ydy)\wedge\dot\eta=0.
$$
The hypothesis $\dot{q}^2(\dot{q}_x\dot{q}_x+\dot{q}_y\dot{q}_y)>0$ implies that the shifted time $s:=t-1/p$ is a local coordinate. The first equation in the system~($*$) is therefore equivalent to $Cp^2=\exp(f(s))$, which gives
$$
C=\exp\big(f(t-1/\dot{q})-2\log|\dot q|\big).
$$
The second equation in the system~($*$) now gives $\dot\eta=h(p_xdx+p_ydy)$ for some smooth function $h$ in two variables. Differentiating again yields
$$
(p_xdx+p_ydy)\wedge(h_xdx+h_ydy)=0.
$$
Therefore, we have that 
$h=\partial H/\partial\dot{q}$
locally for a smooth function $H=H(t,p)$, and consequently
$$
\dot{\eta}=\frac{\partial H(t,\dot{q})}{\partial x}dx+\frac{\partial H(t,\dot{q})}{\partial y}dy.
$$
The converse follows by reversing the computation. Since $F_\omega=d\omega$, the Bianchi identity $dF_\omega=0$ is automatic. This completes the proof.
\end{proof}

We shall now provide a generic local example.
On the half-plane $\mathbb{H}^2$ where $y>0$, straightforward computation verifies that
$$
q:=\log((x+t)^2+y^2)-2\log(y)
$$
is a solution to the Toda equation as $$(e^q)_{tt}=q_{xx}+q_{yy}=\frac{2}{y^2}.$$
We also obtain that
$$
p:=\dot q=\frac{2(x+t)}{(x+t)^2+y^2}
$$
and that
$$
p_xp_x+p_yp_y=\frac{4}{((x+t)^2+y^2)^2}>0.
$$
On the region where $x+t\neq0$, the choices $\eta=0$ and
$$
C:=\frac{1}{p^2}=\frac{((x+t)^2+y^2)^2}{4(x+t)^2}
$$
for the metric 
$$
g=-dt^2+e^q(dx^2+dy^2)+C^2(\eta+dz)^2
$$
correspond to $f\equiv0$ in Theorem~\ref{main_theorem}, and hence give a genuinely time-dependent electromagnetic field with nonzero electric and magnetic parts.

The preceding classification applies where the time derivative of the Toda scalar field has a nonzero horizontal gradient. When this gradient vanishes, the characteristic used in the proof loses rank and Maxwell's equation acquires additional freedom. The separated solutions in the following theorem lie in this complementary branch.

\begin{theorem}\label{pure_magnetic}
Let $K,\beta\in\mathbb{R}$ and $q(t;x,y):=\log(r(t))+2u(x,y)$, where $r(t):=-t^2K+\beta>0$ and $u=u(x,y)$ solves the Liouville equation 
\begin{equation}\label{Liouville}
u_{xx}+u_{yy}=-Ke^{2u}.
\end{equation}
Then, for any smooth functions $\lambda:=\lambda(t),P:=P(t;x,y),Q:=Q(t;x,y)$ satisfying $\lambda>0$ and $Q_x=P_y$, the spin connection $\omega=d^cq$ of the 3+1 spacetime 
$$
g:=-dt^2+e^q(dx^2+dy^2)+\lambda(Pdx+Qdy+dz)^2
$$
is an electromagnetic potential, of which Faraday tensor $F_\omega=Ke^{2u}dx\wedge dy$ is purely magnetic.
\end{theorem}
\begin{proof}
Again, we denote by $\star$ the Hodge star operator of $g$. Since $\ddot r=-2K$, we have that
$$
\ddot q+\dot q^2=\frac{\ddot r}{r}
=-\frac{2K}{r}
=e^{-q}(q_{xx}+q_{yy}),
$$
where the last equality follows from (\ref{Liouville}). Moreover, it holds that
$$
F_\omega=d\omega
=-\frac{1}{2}(q_{xx}+q_{yy})dx\wedge dy
=Ke^{2u}dx\wedge dy,
$$
so the field strength $F_\omega$ is purely magnetic. Writing $\theta:=Pdx+Qdy+dz$, we obtain
$$
\star F_\omega=\frac{K}{r}\sqrt{\lambda}dt\wedge\theta.
$$
As both $\lambda(t)$ and $r(t)$ depend only on time $t$, we conclude that
$$
d\star F_\omega
=-\frac{K}{r}\sqrt{\lambda}dt\wedge(Q_x-P_y)dx\wedge dy=0,
$$
which verifies the Maxwell equation. The Bianchi identity $dF_\omega=0$ is again automatic.
\end{proof}

Thus, the loss of the generic characteristic is not a degeneracy of the electromagnetic field. It replaces the rigid relation by an arbitrary positive time-dependent fibre scale and an arbitrary closed horizontal fibre form.

\section{An Example of Spherical Monopoles}

We now globalize the curvature of the separated magnetic branch. The spacetime interval in this section is defined for time $0<t<1$, while the values $t=0$ and $t=1$ refer only to its limiting ends.

We assume now the holographic screen of the spacetime metric 
$$
g=-dt^2+e^q(dx^2+dy^2)+(Adx+Bdy+Cdz)^2
$$
in (\ref{metric}) is the 2-sphere, meaning that the $xOy$-plane is a stereographic coordinate chart of $\mathbb{S}^2$. We then choose the constants in Theorem~\ref{pure_magnetic} to be $K=\beta=1$,
so that $r(t):=1-t^2$ is positive for $t\in[0,1)$. Consider the elliptic solution $u=u(x,y)$ to the Liouville equation~(\ref{Liouville}) given by
$$
e^u:=\frac{2}{1+x^2+y^2}
$$
and write $q(t;x,y):=\log(r(t))+2u(x,y)$ as in Theorem~\ref{pure_magnetic}. Then, the 3+1 spacetime metric 
$$
g:=-dt^2+e^q(dx^2+dy^2)+t^{-4}(1-t^2)^4dz^2
$$
satisfies all the hypotheses of Theorem~\ref{pure_magnetic}, so its spin connection
$$
\omega:=d^cq=2\frac{xdy-ydx}{1+x^2+y^2}
$$
is an electromagnetic potential, of which Faraday tensor $$
F_\omega=Ke^{2u}dx\wedge dy=\frac{4dx\wedge dy}{(1+x^2+y^2)^2}
$$ 
is purely magnetic. In particular, the spin connection $\omega=d^cq$ of the spacetime $g$ enjoys unit magnetic charge
$$
\frac{1}{\pi}\iint_{\mathbb{R}^2}\frac{dx\wedge dy}{(1+x^2+y^2)^2}=\frac{1}{4\pi}\oiint_{\mathbb{S}^2}Ke^{2u}dx\wedge dy=\frac{1}{2}\chi(\mathbb{S}^2)=1
$$
by the Gauss-Bonnet theorem. Here, the $\operatorname{U}(1)$ connection 1-form $\omega$ is a local gauge potential on the stereographic coordinate chart, whereas its curvature is global on the sphere $\mathbb{S}^2$.

This is precisely the usual Dirac--Wu--Yang distinction between a local gauge potential and a global monopole curvature described in \cite{Dirac,WuYang}, now placed on a time-dependent curved background. The normalization above is the physical magnetic charge. As the curvature of the oriented screen frame bundle, the same field has Chern number
$$
\frac{1}{2\pi}\oiint_{\mathbb{S}^2}F_\omega=2.
$$

Finally, a computation also yields that the magnetic field generated by the potential $\omega$ has magnitude $(1-t^2)^{-1}$ and is pointing in the direction of the positive $z$-axis, so that its flux is conserved while the measured magnetic field grows as the holographic screen 2-sphere contracts over time.

We also remark that $t=0$ is a monopole-zero end, and that $t=1$ is a contracting singular end of the universe.
Here, the \emph{monopole-zero} refers to the Toda monopole scalar $\dot q$, which tends to zero as $t\to0$, while the Faraday curvature 2-form itself remains nonzero. At the same end, the Killing-fibre radius diverges, while at $t=1$ both the screen and the fibre collapse.

\section{Concluding Remarks}

The central point of the construction is the coexistence of two classical first-order structures. The Toda equation makes the pair $(\dot q,2d^cq)$ a generalized Abelian monopole on the $2+1$ dimensional quotient manifold, while the 4-dimensional Maxwell equation tests the same connection, but against a Lorentzian Hodge star. On the generic branch, the second requirement is equivalent to the reciprocal characteristic formula for $C$ and to the level-set condition for $\dot\eta$. The separated branch has different rank and yields the purely magnetic Liouville family. The spherical member has conserved flux and unit physical charge.

The result should not be confused with the Jones--Tod self-dual lifting: In that construction, the monopole connection controls the 4-dimensional fibration, and a further monopole can produce a self-dual electromagnetic field. Here, the Killing-fibre connection is independent, the normalized time-autonomous member has $C^2=1/\dot q^4$, and the resulting Lorentzian background is allowed to be neither Einstein nor self-dual. The contribution is therefore the exact test-field realization and its local exact classification, not a new Toda equation or a new generalized monopole.

Several extensions remain natural. General Toda monopoles satisfy a linear equation over the same Einstein--Weyl background, and it is reasonable to classify which of them admit analogous Lorentzian electromagnetic lifts. The known hydrodynamic and inverse spectral constructions of Toda fields \cite{Ferapontov,ManakovSantini,DunajskiFerapontovKruglikov,CalderbankKruglikov} should supply further explicit electromagnetic backgrounds. Global versions must compare the Killing circle bundle with the $\operatorname{U}(1)$ bundle and control the energy, as well as the flux quantization. Gravitational back-reaction and non-Abelian analogues are separate problems and are best left beyond the present test-field theory.

\section*{Statements and Declarations}

No funding was received to assist with the preparation of this manuscript.

The author certifies that the author has no affiliations with or involvement in any other organization or entity with any financial interest or non-financial interest in the subject matter or materials discussed in this manuscript.

Data sharing is not applicable to this article as no datasets were generated or analyzed during the current study.

\singlespacing
\sloppy


\begin{thebibliography}{99}

\bibitem{Dirac}
P.~A.~M. Dirac,
\emph{Quantised singularities in the electromagnetic field},
Proc. Roy. Soc. Lond. A \textbf{133} (1931), 60--72,
\href{https://doi.org/10.1098/rspa.1931.0130}{doi:10.1098/rspa.1931.0130}.

\bibitem{WuYang}
T.~T. Wu and C.~N. Yang,
\emph{Concept of nonintegrable phase factors and global formulation of gauge fields},
Phys. Rev. D \textbf{12} (1975), 3845--3857,
\href{https://doi.org/10.1103/PhysRevD.12.3845}{doi:10.1103/PhysRevD.12.3845}.

\bibitem{GibbonsHawking}
G.~W. Gibbons and S.~W. Hawking,
\emph{Gravitational multi-instantons},
Phys. Lett. B \textbf{78} (1978), 430--432,
\href{https://doi.org/10.1016/0370-2693(78)90478-1}{doi:10.1016/0370-2693(78)90478-1}.

\bibitem{BoyerFinley}
C.~P. Boyer and J.~D. Finley III,
\emph{Killing vectors in self-dual, Euclidean Einstein spaces},
J. Math. Phys. \textbf{23} (1982), 1126--1130,
\href{https://doi.org/10.1063/1.525479}{doi:10.1063/1.525479}.

\bibitem{BoyerPlebanski}
C.~P. Boyer and J.~F. Pleba\'nski,
\emph{Conformally self-dual spaces and Maxwell's equations},
Phys. Lett. A \textbf{106} (1984), 125--129,
\href{https://doi.org/10.1016/0375-9601(84)90904-6}{doi:10.1016/0375-9601(84)90904-6}.

\bibitem{JonesTod}
P.~E. Jones and K.~P. Tod,
\emph{Minitwistor spaces and Einstein--Weyl spaces},
Class. Quantum Grav. \textbf{2} (1985), 565--577,
\href{https://doi.org/10.1088/0264-9381/2/4/021}{doi:10.1088/0264-9381/2/4/021}.

\bibitem{SavelievVershik}
M.~V. Saveliev and A.~M. Vershik,
\emph{Continuum analogues of contragredient Lie algebras},
Commun. Math. Phys. \textbf{126} (1989), 367--378,
\href{https://doi.org/10.1007/BF02125130}{doi:10.1007/BF02125130}.

\bibitem{Ward}
R.~S. Ward,
\emph{Einstein--Weyl spaces and $\operatorname{SU}(\infty)$ Toda fields},
Class. Quantum Grav. \textbf{7} (1990), L95--L98,
\href{https://doi.org/10.1088/0264-9381/7/4/003}{doi:10.1088/0264-9381/7/4/003}.

\bibitem{TakasakiTakebe}
K. Takasaki and T. Takebe,
\emph{$\operatorname{SDiff}(2)$ Toda equation---hierarchy, tau function, and symmetries},
Lett. Math. Phys. \textbf{23} (1991), 205--214,
\href{https://doi.org/10.1007/BF01885498}{doi:10.1007/BF01885498}.

\bibitem{Racz}
I. R\'acz,
\emph{electromagnetic fields in spacetimes admitting non-null Killing vectors},
Class. Quantum Grav. \textbf{10} (1993), L167--L172,
\href{https://arxiv.org/abs/gr-qc/9303014}{arXiv:gr-qc/9303014}.

\bibitem{Strachan}
I.~A.~B. Strachan,
\emph{The dispersive self-dual Einstein equations and the Toda lattice},
J. Phys. A \textbf{29} (1996), 6117--6124,
\href{https://doi.org/10.1088/0305-4470/29/18/036}{doi:10.1088/0305-4470/29/18/036}.

\bibitem{Calderbank}
D.~M.~J. Calderbank,
\emph{The geometry of the Toda equation},
J. Geom. Phys. \textbf{36} (2000), 152--162,
\href{https://doi.org/10.1016/S0393-0440(00)00019-X}{doi:10.1016/S0393-0440(00)00019-X}.

\bibitem{CalderbankPedersen}
D.~M.~J. Calderbank and H. Pedersen,
\emph{Selfdual spaces with complex structures, Einstein--Weyl geometry and geodesics},
Ann. Inst. Fourier \textbf{50} (2000), 921--963,
\href{https://doi.org/10.5802/aif.1779}{doi:10.5802/aif.1779}.

\bibitem{CalderbankTod}
D.~M.~J. Calderbank and P. Tod,
\emph{Einstein metrics, hypercomplex structures and the Toda field equation},
Differential Geom. Appl. \textbf{14} (2001), 199--208,
\href{https://doi.org/10.1016/S0926-2245(01)00037-7}{doi:10.1016/S0926-2245(01)00037-7}.

\bibitem{Ferapontov}
E.~V. Ferapontov, D.~A. Korotkin and V.~A. Shramchenko,
\emph{Boyer--Finley equation and systems of hydrodynamic type},
Class. Quantum Grav. \textbf{19} (2002), L205--L210,
\href{https://doi.org/10.1088/0264-9381/19/24/101}{doi:10.1088/0264-9381/19/24/101}.

\bibitem{ManakovSantini}
S.~V. Manakov and P.~M. Santini,
\emph{The dispersionless 2D Toda equation: dressing, Cauchy problem, longtime behavior, implicit solutions and wave breaking},
J. Phys. A \textbf{42} (2009), 095203,
\href{https://doi.org/10.1088/1751-8113/42/9/095203}{doi:10.1088/1751-8113/42/9/095203}.

\bibitem{DunajskiFerapontovKruglikov}
M. Dunajski, E.~V. Ferapontov and B. Kruglikov,
\emph{On the Einstein--Weyl and conformal self-duality equations},
J. Math. Phys. \textbf{56} (2015), 083501,
\href{https://doi.org/10.1063/1.4927251}{doi:10.1063/1.4927251}.

\bibitem{CalderbankKruglikov}
D.~M.~J. Calderbank and B. Kruglikov,
\emph{Integrability via geometry: dispersionless differential equations in three and four dimensions},
Commun. Math. Phys. \textbf{382} (2021), 1811--1841,
\href{https://doi.org/10.1007/s00220-020-03913-y}{doi:10.1007/s00220-020-03913-y}.

\bibitem{Araneda}
B. Araneda,
\emph{Hidden symmetries of generalised gravitational instantons},
Ann. Henri Poincar\'e \textbf{26} (2025), 4021--4049,
\href{https://doi.org/10.1007/s00023-024-01515-1}{doi:10.1007/s00023-024-01515-1}.

\end{thebibliography}
\end{document}